\documentclass{article}
\usepackage{format}
\usepackage[linesnumbered,ruled,vlined]{algorithm2e}

\def\SOdd{S_{\textup{odd}}}
\def\SEven{S_{\textup{even}}}
\def\muOdd{\mu_{\textup{odd}}}
\def\muEven{\mu_{\textup{even}}}

\title{A tight lower bound on non-adaptive group testing estimation}

\author{
Nader H. Bshouty\thanks{Department of Computer Science, Technion. Email: \texttt{bshouty@cs.technion.ac.il}}
\and
Tsun-Ming Cheung\thanks{School of Computer Science, McGill University. Email: \texttt{tsun.ming.cheung@mail.mcgill.ca}}
\and
Gergely Harcos\thanks{Number Theory Division, Alfr\'ed R\'enyi Institute of Mathematics. Supported by the R\'enyi Int\'ezet Lend\"ulet Automorphic Research Group and NKFIH (National Research, Development and Innovation Office) grant K~143876. Email \texttt{harcos.gergely@renyi.hu}}
\and
Hamed Hatami\thanks{School of Computer Science, McGill University. Research supported by an NSERC grant. Email: \texttt{hatami@cs.mcgill.ca}} 
\and
Anthony Ostuni\thanks{Department of Computer Science and Engineering, UC San Diego. Supported by NSF award 2006443. Email: \texttt{aostuni@ucsd.edu}}
}
\date{}

\begin{document}

\maketitle

\begin{abstract}
    Efficiently counting or detecting defective items is a crucial task in various fields ranging from biological testing to quality control to streaming algorithms. The \emph{group testing estimation problem} concerns estimating the number of defective elements $d$ in a collection of $n$ total within a given factor. We primarily consider the classical query model, in which a query reveals whether the selected group of elements contains a defective one. We show that any non-adaptive randomized algorithm that estimates the value of $d$ within a constant factor requires $\Omega(\log n)$ queries. This confirms that a known $O(\log n)$ upper bound by Bshouty (2019) is tight and resolves a conjecture by Damaschke and Sheikh Muhammad (2010). Additionally, we prove similar matching upper and lower bounds in the threshold query model.
\end{abstract}

\section{Introduction}
The \emph{group testing} problem is a fundamental computational problem that concerns counting or detecting a set of defective items. Suppose there are $d$ defective items in a collection of $n$ total. Under the group testing model, a query can determine whether a selected group of elements contains a defective one. The concept was pioneered by Dorfman \cite{dorfman1943detection} to obtain more efficient syphilis testing by mixing blood samples together to perform tests in groups, and since then has exploded into a richly studied area. There are now various models considered \cite{damaschke2010bounds, atia2012boolean, cheraghchi2013noise, aldridge2019group, bshouty2019lower, bshouty2020optimal}, each with different applications as diverse as DNA testing \cite{gille1991pooling, balding1996comparative, curnow1998pooling}, learning theory \cite{gilbert2008group, malioutov2013exact, emad2015semiquantitative, malioutov2017learning}, and industrial processes \cite{sobel1959group}. Readers may see \cite[Section 1.7]{aldridge2019group} for a more extensive discussion. 

In this work, we focus on \emph{estimating} the number of defective items $d$ in the randomized estimation setting. A randomized algorithm makes queries that may depend on randomness and outputs a valid answer with probability at least 2/3. For a given factor $\alpha = 1+\Omega(1)$, we say $d^*$ is an $\alpha$-estimation for $d$ if $d\leq d^*\leq \alpha d$.

An important distinction for query problems is the \emph{adaptiveness} of algorithms.
Adaptive algorithms may use the results of prior queries to choose the future ones. Nearly-tight bounds are known in the adaptive randomized estimation setting: Bshouty et al.\ \cite{bshouty2018adaptive} provided a randomized algorithm that returns a value $d^*$ with $(1-\eps)d \le d^* \le (1+\eps)d$, and makes at most $\log \log d + \log^*n + O\left(1/\eps^2\right)$ queries, where $\log^*$ is the iterated logarithm function.
They also gave a nearly matching query complexity lower bound of $\log \log d + \Omega\left(1/\eps\right)$.

Our work considers the \emph{non-adaptive} case, where all queries are chosen at once. Despite being less powerful, non-adaptive algorithms have the benefits of parallelized testing and simpler test designs, which are significant for efficient implementation in real-world applications. Bshouty provided a polynomial-time $O(\log n)$ query constant-estimation algorithm \cite{bshouty2019lower} using ideas similar to \cite{damaschke2010competitive, falahatgar2016estimating}, and it was conjectured to be tight \cite{damaschke2010bounds}. A close lower bound of $\Omega(\log n / \log \log n)$ was independently proven in \cite{ron2016power} and \cite{bshouty2019lower}. Very recently, Bshouty posted a manuscript that established the lower bound $\Omega(\log n / \log \log \stackrel{k}{\cdots} \log n)$ for any constant $k$ \cite{bshouty2023improved}; this is the same bound errantly claimed earlier \cite{bshouty2018lower}. In this work, we confirm that the conjectured $\Omega(\log n)$ bound is indeed tight.
\begin{theorem}\label{thm:main}
    Let $\alpha = 1+\Omega(1)$. Any non-adaptive randomized algorithm that $\alpha$-estimates the number of defective items in a set of $n$ total must perform $\Omega(\log_\alpha n)$ queries.
\end{theorem}
Note \Cref{thm:main} holds for any estimation factor $\alpha=1+\Omega(1)$, even in super-constant regimes. For example, it states that a $\log n$-estimation algorithm would require $\Omega(\log n/\log\log n)$ queries, illustrating that $\alpha$-estimation is roughly as difficult as constant estimation unless $\alpha$ is quasi-polynomial in $n$. It should be noted, however, that this lower bound already follows from the techniques of \cite{bshouty2019lower} and \cite{bshouty2023improved} when $\alpha = \Omega(\log n)$ and $\alpha = \Omega(\log \log \cdots \log n)$, respectively.

The main result follows from \cref{thm:threshold} as a special case under the \emph{threshold query} model in the \emph{promise} setting. Introduced in \cite{damaschke2006threshold} and well-studied since then \cite{ADL11,He2012,BCE20,CJZ23}, the threshold query model is defined in the absolute sense: for a fixed threshold parameter $\lambda\in [n]$, the query oracle returns 1 on a set of items $Q$ if and only if $Q$ contains at least $\lambda$ defective items. Unlike some prior works, we do not distinguish between constant and non-constant thresholds in our lower bound. In the threshold model framework, the classical query oracle is simply a 1-threshold.

Notice that the threshold query model cannot produce a good prediction when the number of defective items is below the threshold. Specifically for $\lambda\geq 3$, the $\lambda$-threshold query result on a collection with 1 or $\lambda-1$ defective items is identically 0 regardless of the query size. To exclude the vacuous cases, it is reasonable to assume that the collection contains at least $\lambda$ defective items. More generally, for non-negative integers $L$ and $U$ with $L<U$, the \emph{$[L,U]$-promise} states that the number of defective items $d$ satisfies $L\leq d\leq U$.

We show a logarithmic query complexity lower bound in terms of the promise gap for the promise estimation problem with threshold queries. The lower bound for the standard query model follows by setting $\lambda=1$, $L=1$, and $U=n$, which attains the claimed bound of $\Omega(\log_\alpha n)$.
\begin{theorem}\label{thm:threshold}
    Let $\alpha=1+\Omega(1)$ and $\lambda\in [n]$. Suppose $L$ and $U$ are integers that satisfy $\lambda\leq L<U\leq n$. Under the $[L,U]$-promise, any non-adaptive randomized $\alpha$-estimation algorithm with $\lambda$-threshold queries must perform $\Omega(\log_\alpha (U/L))$ queries.
\end{theorem}
In \Cref{Upper_Bound} we prove this result is tight for any $\alpha$ and constant $\lambda$ by extending the upper bound in \cite{bshouty2019lower} to the general threshold case.

\subsection{Distribution distinguishing problem}\label{sec:ddist}
A crucial idea of both prior work \cite{bshouty2019lower, bshouty2023improved} and our proof is the connection with the \emph{distribution distinguishing} problem, an algorithmic formulation of hypothesis testing. The general setting of the problem is as follows: given two probability distributions $\mu_0$ and $\mu_1$, a uniform bit $b\in\set{0,1}$ is chosen, then a sample $x$ is drawn from the distribution $\mu_b$. The objective is to determine $b$ from the provided sample $x$. Intuitively, the distribution distinguishing problem is hard if the two distributions are close, in the sense that $\mu_0(A)\approx \mu_1(A)$ for every subset $A$ in the probability space.

The earlier works of \cite{bshouty2019lower, bshouty2023improved} make use of the distribution distinguishing problem by showing that if too few queries are made, a vacant set of queries would be able to distinguish between two distributions. More precisely, for a non-adaptive, $\alpha$-estimation algorithm $\mathcal{A}$, they partition the random $q$ queries of $\mathcal{A}$ into $b$ buckets depending on the query sizes, and show that there exists a specific (i.e.,\ the choice does not depend on the randomness) bucket $\mathcal{B}$ that contains $O(q/b)$ queries with high probability. Moreover, $\mathcal{B}$ determines two distributions on sets of defective items that $\mathcal{A}$ can distinguish between, but where (essentially) the only informative queries are in $\mathcal{B}$. If $q = o(b)$, then $\mathcal{B}$ will typically contain zero queries, so $\mathcal{A}$ cannot distinguish between distributions, leading to a contradiction. The earlier paper \cite{bshouty2019lower} carries out this strategy with $b = \Theta(\log n/\log\log n)$, whereas the recent work \cite{bshouty2023improved} performs a more sophisticated iteration process to obtain the strengthened bound $\Omega(\log n / \log \log \stackrel{k}{\cdots} \log n)$ for any constant $k$.

Our proof circumvents the need to design bucket-dependent distributions and further complications, and simply adopts a pair of distributions ``naturally distinguishable'' by an $\alpha$-estimation algorithm.
We start by constructing a pair of distributions $(\muEven,\muOdd)$ on the planted set of defective items with two properties. The first is that the distributions should have sufficiently disjoint supports. That is, no element sampled from $\muEven$ can have a size within an $\alpha$-factor of the size of an element sampled from $\muOdd$. Thus, any algorithm which can $\alpha$-estimate the number of defective items must be able to distinguish between the two distributions. The second is that without making a large number of queries, the distributions \emph{induced by the queries} should be difficult to distinguish. Intuitively, the hardness requirement is related to how close the induced distributions are, and this is formally quantified by the total variation distance. 

More precisely, we choose the distributions to ``multiplicatively interleave'' the possible support sizes, where $\muEven$ is a distribution over sets of size even powers of $\lfloor \alpha \rfloor + 1$ and $\muOdd$ over odd powers. Observe we are essentially viewing the number of defective items on a logarithmic scale. In order to correctly approximate this number, any non-adaptive algorithm must choose the appropriate scale; otherwise, the query results will be identically 0 or 1 with high probability.  We randomly plant defective items such that if $o(\log n)$ queries are made, it is unlikely the correct scale is chosen, and thus little information can be learned.

Our analysis relies on the powerful \emph{coupling} technique, which we believe may have more applications for related problems. We will formalize the notions of distribution distance and query-induced distributions in \cref{sec:prelims}.

\paragraph{Concurrent work} This paper is a combination of independent and concurrent work by Bshouty and Harcos \cite{bshouty2023tight} and Cheung, Hatami, and Ostuni \cite{cheung2023tight}, both of which proved \cref{thm:main} using similar techniques.

\paragraph{Overview} 
In \cref{sec:prelims}, we lay out several tools and technical estimates that are required to prove our results, including the notion of query-induced distributions in \cref{subsec:notations}, a refresher on the hypergeometric distribution in \cref{subsec:hypergeom}, and fundamental facts about total variation distance and coupling in \cref{subsec:TVD}. \cref{sec:proof} contains the proof of \cref{thm:threshold}. 
We conclude with some remarks in \Cref{sec:conclude}. \Cref{Upper_Bound} contains a tight upper bound in the threshold query model, while Appendices~\ref{CouplingBC} and \ref{EandToDet} provide a more thorough background on coupling and randomized to deterministic conversions, respectively.

\section{Preliminaries}\label{sec:prelims}

\subsection{Notations}\label{subsec:notations}
We denote $[n]$ for $\set{1,\ldots,n}$, and adopt the standard computer science asymptotic notations. 

We introduce the following notations for the threshold query model.
In the non-adaptive setting, $q$ queries are specified by a collection of subsets $Q_1,\ldots,Q_q\subseteq [n]$, such that $Q_i$ does not depend on the responses to other queries. For a set of defective items $B\subseteq[n]$ and the query set $W\subseteq [n]$, we use the notation $W^{\geq \lambda}(\cdot)$ to denote the query result on $W$:
\[W^{\geq \lambda}(B)=\begin{cases}
    1 &\text{ if }\abs{W\cap B}\geq \lambda\\
    0 &\text{ otherwise}
\end{cases}.\]
For $\vec{Q}=(Q_1,\ldots,Q_q)$, we use the shorthand $\vec{Q}^{\geq \lambda}(B)=(Q_1^{\geq \lambda}(B),\ldots,Q_q^{\geq \lambda}(B))\in\set{0,1}^q$ to denote the tuple of $q$ queries. For a distribution $\mu$ on subsets of $[n]$, the \emph{induced distribution} is the distribution of $\vec{Q}^{\geq \lambda}(B)$ for $B$ sampled from $\mu$.

\subsection{Hypergeometric distribution}\label{subsec:hypergeom}
The \emph{hypergeometric distribution} $\mathcal{H}_{n,k,s}$ is naturally associated with the group testing problem. It is characterized by three parameters:
\begin{itemize}
    \item $n$: the total number of items;
    \item $k$: the number of marked items;
    \item $s$: the number of items sampled in a uniform draw (without replacement).
\end{itemize}
The number of marked items sampled in the draw is given by the distribution:
\[\Pr(r\text{ marked items sampled})=\frac{\binom{k}{r}\binom{n-k}{s-r}}{\binom{n}{s}}.\]
We adopt the convention that $\binom{N}{R}=0$ whenever $R>N$ or $R<0$.

It is a well-known fact that the hypergeometric random variable $H_{n,k,s}\sim \mathcal{H}_{n,k,s}$ has mean $ks/n$. For the proof of the main result, we will need the following concentration inequalities for hypergeometric distributions. Markov's inequality implies that
\begin{equation}
    \Pr(H_{n,k,s} \geq \gamma) \leq \frac{ks}{\gamma n}. \label{eq:markov}
\end{equation}
A Chernoff-type lower tail bound is implicit in \cite{mulzer2019proofs}: for $\xi<ks/n$,
\begin{equation}
    \Pr\left(H_{n,k,s} \leq \xi\right) \leq \exp\left(-\frac{(ks/n-\xi)^2}{2ks/n}\right).\label{eq:chernoff}
\end{equation}

\subsection{Total variation distance and coupling}\label{subsec:TVD}
As mentioned in \cref{sec:ddist}, the key technique of our proof is to construct two close distributions of sets of defective items. The closeness is quantified by the \emph{total variation distance}.
\begin{definition}[Total variation distance] Let $\mu_0$ and $\mu_1$ be two probability measures on the measurable space $(\mathcal{S}, \mathcal{F})$. The \textup{total variation distance} of the two measures is
\begin{equation*}
\|\mu_0 - \mu_1\|_{\TV} \defeq \sup_{A \in \mathcal{F}} |\mu_0(A) - \mu_1(A)|. \label{eq:TV}
\end{equation*}
\end{definition}

The above notion provides a mathematical interpretation of statistical closeness. From the algorithmic perspective, statistical closeness can be captured by the hardness of the distribution distinguishing problem. More precisely, the total variation distance can be viewed as a measure of how well the optimal algorithm for the distinguishing problem (also called \emph{distinguisher}) outperforms a random guess. 

Let $\mu_0$ and $\mu_1$ be two probability measures on the measurable space $(\mathcal{S}, \mathcal{F})$. We say that $\Psi: \mathcal{S}\to\set{0,1}$ is a deterministic distinguisher between $\mu_0$ and $\mu_1$ if 
\begin{eqnarray*}\label{RD-def}
    \Pr_{\substack{b \sim \{0,1\}\\x \sim \mu_b
    }}(\Psi(x)=b)\ge \frac{2}{3}.
\end{eqnarray*}

\begin{lemma}\label{lem:TV-algo}
    Let $\mu_0$ and $\mu_1$ be two probability measures on the measurable space $(\mathcal{S}, \mathcal{F})$. Let $\Psi:\mathcal{S}\to\set{0,1}$ be any deterministic distinguisher between $\mu_0$ and $\mu_1$. Then
    \[\frac{1}{3}\le 2\left(\Pr_{\substack{b \sim \{0,1\}\\x \sim \mu_b}}(\Psi(x)=b)- \frac{1}{2}\right)\leq \norm{\mu_0-\mu_1}_{\TV}.\]
\end{lemma}
\begin{proof}
     Let $A$ be the support of $\Psi$. Then
    \[\frac{1}{6}\le \Pr_{\substack{b \sim \{0,1\}\\x \sim \mu_b}}(\Psi(x)=b)-\frac{1}{2}=\frac{1}{2}\left(\Pr_{x\sim \mu_0}(x \in A)+\Pr_{x\sim \mu_1}(x\notin A)-1\right)=\frac{1}{2}(\mu_0(A)-\mu_1(A))\le \frac{1}{2}\norm{\mu_0-\mu_1}_{\TV}.\]    
\end{proof}

By averaging, the above bound also holds for randomized distinguishers.

Often exactly computing the total variation distance is difficult. A fundamental connection with the notion of \emph{coupling} provides a way to upper bound the total variation distance.
\begin{definition}[Coupling]
Let $\mu_1$ and $\mu_2$ be two probability measures on the measurable space $(\mathcal{S}, \mathcal{F})$. A \textup{coupling} of $\mu_1$ and $\mu_2$ is a probability measure $\zeta$ on the product space $(\mathcal{S}\times \mathcal{S}, \mathcal{F} \times \mathcal{F})$ such that
\[
\zeta(A \times \mathcal{S}) = \mu_1(A) \textup{ and } \zeta(\mathcal{S} \times A) = \mu_2(A) \textup{ for all } A \in \mathcal{F}.
\]
\end{definition}
\begin{theorem}[Coupling inequality]\label{couplingineq}
    Let $\mu_1$ and $\mu_2$ be two probability measures on the measurable space $(\mathcal{S}, \mathcal{F})$. For any coupling $(X,Y)$ of $\mu_1$ and $\mu_2$,
    \[
    \|\mu_1 - \mu_2\|_{\TV} \le \Pr(X \ne Y).
    \]
\end{theorem}
\begin{proof}
For any $A\in\mathcal{F}$, we have
\begin{align*}
    \mu_1(A) - \mu_2(A) &= \Pr(X \in A) - \Pr(Y \in A) \\
    &= \Pr(X \in A, X = Y) + \Pr(X \in A, X \ne Y) - \Pr(Y \in A, X = Y) - \Pr(Y \in A, X \ne Y) \\
    &= \Pr(X \in A, X \ne Y) - \Pr(Y \in A, X \ne Y) \\
    &\le \Pr(X \ne Y). \qedhere
\end{align*}    
\end{proof}
We direct unfamiliar readers to \Cref{CouplingBC} for additional background.

\section{Proof of \cref{thm:threshold}}\label{sec:proof}
This section is dedicated to the proof of \cref{thm:threshold}. To recap the setting, for a fixed $\lambda\in [n]$, the $\lambda$-threshold query oracle detects whether a query set contains at least $\lambda$ defective items. The promise that the collection contains at least $L$ ($\ge \lambda$) and at most $U$ defective items is given. 

We construct a pair of distributions that is hard for any $\alpha$-estimation algorithm to distinguish. Let $\beta=\floor{\alpha}+1$, and define
\begin{align*}
    \SEven &=\set{L\beta^2, L\beta^4,\ldots, L\beta^{2m}},\\
    \SOdd &=\set{L\beta, L\beta^3,\ldots, L\beta^{2m-1}},
\end{align*}
where $m$ is the largest integer which $L\beta^{2m}\leq U$. It is clear that $m=\Theta(\log_\beta (U/L))=\Theta(\log_\alpha (U/L))$.

The distribution $\muEven$ (resp. $\muOdd$) is constructed by the following sampling procedures:
\begin{enumerate}
    \item Sample $s$ from $\SEven$ (resp. $\SOdd$) uniformly at random;
    \item Sample a set of $s$ items uniformly at random.
\end{enumerate}

We claim that if the deterministic queries $\vec{Q}^{\geq\lambda}$ output a valid estimation, the same queries
distinguish $\vec{Q}^{\geq\lambda}(\muOdd)$ and $\vec{Q}^{\geq\lambda}(\muEven)$ without error. Indeed by the choice of $\beta$, there is exactly one $s\in\SOdd\cup\SEven$ which the estimation algorithm output is within the range of $[s,\alpha s]$. Thus \cref{lem:TV-algo} implies 
\begin{equation}\label{eq:constantLower}
    \|\vec{Q}^{\geq\lambda}(\muEven) -\vec{Q}^{\geq\lambda}(\muOdd)\|_{\TV} \ge \frac{1}{3}.
\end{equation}

It remains to derive an upper bound for $\vec{Q}^{\geq\lambda}(\muEven)$ and $\vec{Q}^{\geq\lambda}(\muOdd)$ in terms of the number of queries $q$. In the next lemma, we show that for any deterministic queries, the total variation distance of the induced distributions is proportional to $\frac{1}{\log_\alpha (U/L)}$ and the number of queries.
\begin{lemma}\label{main-lemma}
    Let $\vec{Q}^{\geq\lambda}=(Q_1^{\geq\lambda},\ldots,Q_q^{\geq\lambda})$ be $q$ deterministic queries with the threshold query oracle. Then
    \[\norm{\vec{Q}^{\geq\lambda}(\muEven)-\vec{Q}^{\geq\lambda}(\muOdd)}_{\TV}=O\left(\frac{q}{\log_\alpha (U/L)}\right).\]
\end{lemma}
As a randomized algorithm is simply a distribution over deterministic queries (see \Cref{EandToDet}), \cref{main-lemma} combined with \cref{eq:constantLower} implies any randomized $\alpha$-estimation algorithm using $q$ queries must satisfy
\[ \frac{1}{3}\leq \norm{\vec{Q}^{\geq\lambda}(\muEven)-\vec{Q}^{\geq\lambda}(\muOdd)}_{\TV} \leq O\left(\frac{q}{\log_\alpha (U/L)}\right) \implies q=\Omega(\log_\alpha (U/L)),\]
and this completes the proof.

\begin{proof}[Proof of \cref{main-lemma}]
    We consider the coupling $(X,Y)$ of $\muEven$ and $\muOdd$ defined by the following sampling procedures:
    \begin{enumerate}
        \item Sample $j\in[m]$ uniformly at random;
        \item Sample a size-$L\beta^{2j}$ set uniformly at random to be $X$;
        \item Sample a size-$L\beta^{2j-1}$ set uniformly at random to be $Y$.
    \end{enumerate}
    It is direct to see that $(X,Y)$ is indeed a coupling for $\muEven$ and $\muOdd$; therefore $(\vec{Q}^{\geq\lambda}(X),\vec{Q}^{\geq\lambda}(Y))$ is a coupling for $\vec{Q}^{\geq\lambda}(\muEven)$ and $\vec{Q}^{\geq\lambda}(\muOdd)$ (see \Cref{Composition}). By the coupling inequality (\cref{couplingineq}) and a union bound, we have
    \begin{align*}
        \norm{\vec{Q}^{\geq\lambda}(\muEven)-\vec{Q}^{\geq\lambda}(\muOdd)}_{\TV} 
        \leq \Pr(\vec{Q}^{\geq\lambda}(X)\neq \vec{Q}^{\geq\lambda}(Y))
        \leq \sum^q_{i=1} \Pr(Q_i^{\geq\lambda}(X)\neq Q_i^{\geq\lambda}(Y)).
    \end{align*}
    It suffices to show that 
    \begin{equation}
        \Pr(W^{\geq\lambda}(X)\neq W^{\geq\lambda}(Y))=O\left(\frac{1}{\log_\alpha(U/L)}\right) \label{eq:disagree-prob}
    \end{equation}
    for any $W\subseteq [n]$. The intuition behind the coupling is that the sizes of $X$ and $Y$ are within a $\beta$-factor; consequently one may expect that the threshold query results are most likely equal for defect sets of comparable sizes. Towards proving \cref{eq:disagree-prob}, we actually prove the stronger statement that with high probability, $W^{\geq \lambda}(X)=W^{\geq \lambda}(Y)=0$ when $j$ is small, and $W^{\geq \lambda}(X)=W^{\geq \lambda}(Y)=1$ when $j$ is large. 

    The disagreement probability is an average of conditional probability in terms of $j$:
    \[\Pr(W^{\geq\lambda}(X)\neq W^{\geq\lambda}(Y))=\frac{1}{m}\sum_{j=1}^m \Pr(\set{W^{\geq\lambda}(X),W^{\geq\lambda}(Y)}=\set{0,1}~|~\abs{X}=L\beta^{2j},\abs{Y}=L\beta^{2j-1}).\]
    Denote the conditional probability by $P_j$, and let $k=\abs{W}$.  $P_j$ can be written in terms of two independent hypergeometric random variables:
    \begin{align}
        P_j &=\Pr(H_{n,k,L\beta^{2j}}\geq \lambda)\Pr(H_{n,k,L\beta^{2j-1}}<\lambda)+\Pr(H_{n,k,L\beta^{2j}}<\lambda)\Pr(H_{n,k,L\beta^{2j-1}}\geq \lambda)~\reflectbox{$\defeq$}~P_j^{(1)}+P_j^{(2)}. \label{eq:prob}
    \end{align}

    Let $m_*=m_*(k)\defeq \max\set{j\in\Z:~kL\beta^{2j}/n\leq \lambda}$.     
    We split the sum into the following three ranges:
    \def\JLow{J_{\text{Low}}}
    \def\JMid{J_{\text{Mid}}}
    \def\JHigh{J_{\text{High}}}
    \[\JLow=(-\infty,m_*]\cap [m],\qquad \JMid=\set{m_*+1}\cap[m],\qquad\JHigh=[m_*+2,\infty)\cap[m].\]
    We want to show that
    \[\sum_{j=1}^m P^{(1)}_j=\sum_{j\in \JLow} P^{(1)}_j+\sum_{j\in \JMid} P^{(1)}_j+\sum_{j\in\JHigh} P^{(1)}_j =O(1),\]
    and the sum of $P^{(2)}_j$ can be handled by a similar argument.    
    It is possible that some of these index sets are empty, in which case one can simply omit the empty sums. For the rest of the proof, we assume that all three sets are non-empty.
     
    For $j\in \JLow$, we use Markov's inequality (\cref{eq:markov}):
    \begin{equation*}
        P^{(1)}_j=\Pr(H_{n,k,L\beta^{2j}}\geq \lambda)\Pr(H_{n,k,L\beta^{2j-1}}< \lambda) \leq \Pr(H_{n,k,L\beta^{2j}}\geq \lambda)\leq \frac{kL\beta^{2j}}{\lambda n}. 
    \end{equation*}
    From this, we can show that the contribution for $j\in \JLow$ is constant:
    \[\sum_{j\in \JLow} P^{(1)}_j \leq \frac{kL}{\lambda n} \sum_{j=1}^{m_*} \beta^{2j}= \frac{kL}{\lambda n} \cdot\frac{\beta^2(\beta^{2m_*}-1)}{\beta^2-1}\leq  O(1)\cdot\frac{kL\beta^{2m_*}}{\lambda n}=O(1).\]
    
    For $j\in\JHigh$, we use the Chernoff-type bound (\cref{eq:chernoff}):
    \begin{align*}
        P^{(1)}_j=\Pr(H_{n,k,L\beta^{2j}}\geq \lambda)\Pr(H_{n,k,L\beta^{2j-1}}<\lambda)
        \leq \Pr(H_{n,k,L\beta^{2j-1}}\leq \lambda) 
        \leq \exp\left(-\frac{(kL\beta^{2j-1}/n-\lambda)^2}{2kL\beta^{2j-1}/n}\right). 
    \end{align*}
    By the definition of $m_*$ and the fact that $\beta\geq 2$, we have $\frac{kL}{n}\beta^{2(m_*+2)-1}= \frac{kL}{n}\beta^{2m_*+3}>\beta\lambda\geq 2\lambda$. It is direct to check that for a fixed $\xi$, $\frac{(x-\xi)^2}{x}\geq \frac{x}{4}$ whenever $x\geq 2\xi$. 
    
    Applying the simplified bound, we can show that the contribution for $j\in\JHigh$ is also constant:
    \[\sum_{j\in\JHigh} P^{(1)}_j \leq \sum_{j=m_*+2}^{m} \exp\left(-\frac{1}{8}\cdot \frac{kL\beta^{2j-1}}{n}\right) \leq \sum^\infty_{r=0} \exp\left(-\frac{\lambda\beta^{2r}}{4}\right) \leq \sum^\infty_{r=1} (e^{-\lambda/4})^r=O(1).\]
    For $j\in\JMid$, we use the trivial bound $P^{(1)}_j\leq 1$.    
    Combining with an analogous argument for $P^{(2)}_j$, we have shown that
    \[\Pr(W^{\geq\lambda}(X)\neq W^{\geq\lambda}(Y))=\frac{1}{m}\sum^m_{j=1} (P^{(1)}_j+P^{(2)}_j) = O\left(\frac{1}{m}\right)=O\left(\frac{1}{\log_\alpha (U/L)}\right). \qedhere\]
\end{proof}

\section{Concluding remarks}\label{sec:conclude}
This work illustrates the power of utilizing the distribution distinguishing problem as a lower bound technique for the group testing estimation problem. The main result uses a fairly straightforward coupling, which simply matches the planted defect set sizes within a $\beta$-factor, and this coupling already yields the desired tight lower bound for our case. Therefore it is reasonable to assert that the potential of this technique is not fully explored. One possible future direction is to extend this technique to prove lower bounds for other query models such as gap-threshold queries \cite{damaschke2010bounds} and density tests \cite{ADL11}.

\section{Acknowledgments}
This work was done in part during TSC, HH, and AO's visit at the Simons Institute for the Theory of Computing. We thank anonymous reviewers for useful feedback on an earlier version of this paper.
    
\bibliographystyle{alpha}  
\bibliography{ref} 

\appendix
\section{The Upper Bound}\label{Upper_Bound}
In \cite{bshouty2019lower}, Bshouty presented a polynomial-time $O(\log n)$ query constant-estimation algorithm using ideas similar to \cite{damaschke2010competitive, falahatgar2016estimating}. In this section, we expand upon this result to cover any $\alpha$-estimation and the threshold query model for any $\alpha$ and constant $\lambda$. Specifically, we prove:

\begin{theorem}\label{thm:thresholdUpper}
    Let $\lambda\in \mathbb{N}$ be a constant and $\alpha=1+\Omega(1)$. Suppose $L$ and $U$ are integers that satisfy $\lambda\leq L<U\leq n$. Under the $[L,U]$-promise, there is a non-adaptive randomized $\alpha$-estimation algorithm that makes $\Omega(\log_\alpha (U/L))$ $\lambda$-threshold queries.
\end{theorem}

\subsection{Definitions and Preliminary Results}
In this section, we give some definitions and results we will need to prove Theorem~\ref{thm:thresholdUpper}.

A $\lambda$-threshold {\it $p$-query} is a query $Q$ that contains each $i\in [n]$ randomly and independently with probability~$p$.
For any constant $\lambda$ and given a $\lambda$-threshold $p$-query $Q^{\ge \lambda}$, we define
\begin{eqnarray}\label{PLam}
    P_\lambda(d,p):=\Pr[Q^{\ge \lambda}(I)=1]=1-\sum_{i=0}^{\lambda-1} {d\choose i}p^i(1-p)^{d-i}=\sum_{i=\lambda}^{d} {d\choose i}p^i(1-p)^{d-i},
\end{eqnarray} where $|I|=d$.
For $\lambda=1$, we define $$P(d,p):=P_1(d,p)=\Pr[Q(I)=1]=1-(1-p)^d,$$ where $Q$ is a $p$-query ($1$-threshold $p$-query) and $|I|=d$. 
The following lemma enables us to assume that $d>d'$ for any constant $d'$ independent of $n$.
\begin{lemma}\label{Asump_d}
    Let $d'=O_n(1)$ be a constant. Any algorithm that $\alpha$-estimates $d$ assuming $d>d'$ with $O(\log(1/\delta)\log n)$ queries can be modified to an algorithm that $\alpha$-estimates $d$ for any $d$ with $O(\log(1/\delta)\log n)$ queries.
\end{lemma}

\begin{proof} Let $A$ be an algorithm that, with probability at least $1-\delta/3$, $\alpha$-estimates $d$ assuming $d>d'$. In the proof of Theorem~\ref{thm:thresholdUpper}, we demonstrate the existence of a constant $c$ such that $P_\lambda(d',\lambda/(cd'))-P_\lambda(d',\lambda/(c(2d'))$ is constant. Therefore, by Chernoff bound, we can augment $A$ with
$O(\log( 1/\delta))$ queries of a non-adaptive algorithm that, with probability at least $1-\delta/3$, accepts if $d\le d'$ and rejects if $d\ge 2d'$. Additionally, we incorporate into $A$ the queries from the non-adaptive algorithm that, with probability at least $1-\delta/3$, identifies all the defective items, assuming their count is less than $2d'$. By~\cite{cheraghchi2013noise}, the number of queries in the latter algorithm is $O(d'^2\log d'\log n)=O(\log n)$. 

If the second algorithm accepts ($d\le 2d'$), we employ the algorithm from~\cite{cheraghchi2013noise} to identify all defective items. In particular, we find $d$ exactly. If it rejects ($d> d'$) we run $A$ to $\alpha$-estimate~$d$. 
\end{proof}
For the proof, we will need the following Chernoff bound
\begin{lemma} \label{Chernoff}
    Let $X_1,X_2,\ldots,X_t$ be independent random variables that takes values in $\{0,1\}$. Let $X=(X_1+X_2+\cdots+X_t)/t$ and $\mu\le E[X]$. Then for any $\Gamma\ge \mu$ we have
    $$\Pr[X\ge \Gamma]\le \left(\frac{e^{1-\frac{\mu}{\Gamma}}\mu}{\Gamma}\right)^{\Gamma t}\le \left(\frac{e\mu}{\Gamma}\right)^{\Gamma t}.$$
\end{lemma}

We will also need the following analytic inequality.
\begin{lemma}\label{Ec}
    For every $c\ge 1$ and $x\ge 2$ we have
    $$0\le e^{-1/c}-\left(1-\frac{1}{cx}\right)^x\le \frac{A}{x},$$ where $A=5e^{-1/c}/c^2$.
\end{lemma}
\begin{proof} Let $f(x)=(1-{1}/({cx}))^x$. Then $\lim_{x\to\infty}f(x)=e^{-1/c}$ and $$f'(x)=f(x)\left(\ln(1-\frac{1}{cx})+\frac{1}{cx-1}\right)=f(x)\sum_{i=2}^\infty \frac{i-1}{i(cx)^i} >0.$$ Therefore, $f(x)$ is a strictly monotone increasing function and $0\le e^{-1/c}-f(x)$.
    By the mean value theorem, there exists $\xi \in [1,2]$ for which 
    $$f(2x)-f(x)= f'(\xi x)x=f(\xi x)\left(\ln\left(1-\frac{1}{\xi cx}\right)+\frac{1}{\xi cx-1}\right)x.$$
    We now use the inequalities $\ln(1-y)\le -y+y^2/2$ and $1/(z-1)\le 1/z+2/z^2$ for any $y<1$ and $z>2$ and the fact that $f(x)=(1-1/(cx))^x$ is a strictly monotone increasing function and get
    $$f(2x)-f(x)\le f(\xi x)\left(-\frac{1}{\xi cx}+\frac{1}{2\xi^2c^2x^2}+\frac{1}{\xi cx}+\frac{2}{\xi^2c^2x^2}\right)x\le \frac{2.5f(2x)}{c^2 x}\le \frac{2.5e^{-1/c}}{c^2x}=\frac{A}{2x}.$$
    Therefore 
    $$e^{-1/c}-f(x)=\lim_{n\to\infty} f(2^nx)-f(x)= \lim_{n\to\infty} \sum_{i=1}^n f(2^ix)-f(2^{i-1}x)\le  \lim_{n\to\infty}\sum_{i=1}^n \frac{A}{2^ix}= \frac{A}{x}.$$
\end{proof}

\subsection{Proof of the Theorem}
In this section, we prove Theorem~\ref{thm:thresholdUpper}.
\begin{proof}

We will begin by presenting the proof for the case of $\lambda=1$ and any $\alpha$. The proof concept in~\cite{bshouty2019lower} for $\lambda=1$ and any constant $\alpha$ relies on the following facts:  

    \begin{enumerate}
        \item\label{I1} $P(d,p)$ is a strictly monotone increasing function in $p$.
        \item\label{I2} For some constant $1 \le c =O_{d}(1)$ we have $\Delta(\alpha):=P(d,1/(cd))-P(d,\alpha^{-1/2}/(cd))=O_d(1)$. That is, $\Delta(\alpha)$ is greater than some constant that is independent of $d$ (and $n$).
        \item\label{I25} For every constant $0<\beta=O_d(1)$ there is a constant $b=O_d(1)$ such that for every $d\ge b$ we have $|P(d,1/(cd))-\lim_{x\to\infty}P(x,1/(cx))|\le \beta.$
        \item\label{I3} There is a constant $1>c'=O_d(1)$ such that for every $0\le x\le 1/cd$ we have $P(d,x)\le c'P(d,\alpha^{1/4}x)$. 
    \end{enumerate}

The algorithm's core idea is to estimate $P(d,p)$ at specific geometric progression points  $p=p_i:=\alpha^{i/4}/U$ for $i=0,1,\ldots,4\log_\alpha(U/L)$. We then select the first $i_1$ such that $P(d,p_{i_1})$ closely approximates $P(d,1/(cd))$ and employ $p_{i_1}$ to estimate the value of $d$. Each estimation is constrained to use at most $O(\log(1/\delta))$ queries. To achieve this, we rely on the condition that $P(d,1/(cd))-P(d,1/(\alpha^{1/2}cd))$ is constant (item~\ref{I1}). Since we lack knowledge of the true value of $d$, and therefore of $P(d,1/(cd))$, we substitute it with $P(d',1/(cd'))$, which is in proximity (see (\ref{DeltaEight}) below that follows from item~\ref{I25}). The additional condition in item~\ref{I3} ensures that the initial $i_1$ estimations can be performed with probability at least $1-\delta$ and within the query limit of $O(\log(1/\delta))$ for each estimation.

Using item~\ref{I25} with $\beta=\min(\Delta(\alpha),2)/16$, there is a constant $d'$ such that for every $d\ge d'$ we have
\begin{eqnarray*}
    |P(d,1/(cd))-\lim_{x\to\infty}P(x,1/(cx))|\le \Delta(\alpha)/16.
\end{eqnarray*} Therefore, for every $d\ge d'$ we have
\begin{eqnarray}\label{DeltaEight}
    |P(d',1/(cd'))-P(d,1/(cd))|\le \Delta(\alpha)/8.
\end{eqnarray}

Now, let's proceed with presenting the algorithm and its analysis for estimating $d$, assuming $d$ is greater than a sufficiently large constant $d'$. The result will then follow by applying Lemma~\ref{Asump_d}.
   
The algorithm: Estimate $P(d,\alpha^{i/4}/U)$ for all $i=0,1,\ldots, 4\log_\alpha(U/L)$, each with an additive error of at most $\Delta(\alpha)/8$ using $O(\log(1/\delta))$ queries (for each $i$)\footnote{Claim~\ref{Claim2} shows that this is possible for the first $i_1$ elements of $P(d,\alpha^{i/4}/L)$.}. Select the first $i_1$ for which the estimated value of $P(d,\alpha^{i_1/4}/L)$ is greater than $P(d',1/(cd'))-\Delta(\alpha)/4$ and return $D=U/(c\alpha^{(i_1-1)/4})$.

\begin{algorithm}
\caption{Estimation of $d$}\label{alg:estimate_d}
\KwData{Parameters $\alpha, L, U, \delta, c, d'$ and $\Delta(\alpha)$.}
\KwResult{Estimated value $d\le D\le \alpha d$.}
For $i=0,1,\ldots, 4\log_\alpha(U/L)$\; 
\hskip .2in Estimate $P(d,\alpha^{i/4}/U)$ with an additive error of at most $\Delta(\alpha)/8$ using $O(\log(1/\delta))$ queries\;
Select the first $i_1$ such that the estimated value of $P(d,\alpha^{i_1/4}/U)$ is greater than $P(d',1/(cd'))-\Delta(\alpha)/4$\;
\Return $D=U/(c\alpha^{(i_1-1)/4})$\;
\end{algorithm}

Now, assuming that items~\ref{I1}-\ref{I3} hold, we will proceed to establish the correctness of the algorithm. We prove 

\begin{claim}
   If the first $i_1$ estimations of $P(d,\alpha^{i/4}/U)$ are correct, then $D\in [d, \alpha d]$.
\end{claim}
\begin{proof} Suppose the first $i_1$ estimations are correct.
  Since the estimation of $P(d,\alpha^{i_1/4}/U)$ is greater than $P(d',1/(cd'))-\Delta(\alpha)/4$, we have that $P(d,\alpha^{i_1/4}/U)$ is greater than $P(d',1/(cd'))-3\Delta(\alpha)/8$. 
  Since by (\ref{DeltaEight}), $|P(d',1/(cd'))- P(d,1/(cd))|\le \Delta(\alpha)/8$,  we have that $P(d,\alpha^{i_1/4}/U)>P(d,1/(cd))-\Delta(\alpha)/2$. 
  By item~\ref{I2}, $P(d,\alpha^{i_1/4}/U)>P(d,\alpha^{-1/2}/(cd))$ and therefore by item (\ref{I1}), $\alpha^{i_1/4}/U>\alpha^{-1/2}/(cd)$ and $D=U/(c\alpha^{(i_1-1)/4})<\alpha^{3/4}d\le \alpha d$. 

  If $i_2$ satisfies $\alpha^{1/4}/(cd)\ge \alpha^{i_2/4}/U>1/(cd)$ then by  (\ref{DeltaEight}) and item (\ref{I1}), $P(d,\alpha^{i_2/4}/U)>P(d,1/(cd))\ge P(d',1/(cd'))-\Delta(\alpha)/8$. The estimation of $P(d,\alpha^{i_2/4}/U)$ is greater than $P(d',1/(cd'))-\Delta(\alpha)/4$. Therefore $i_1\le i_2$ and $\alpha^{i_1/4}/U\le \alpha^{i_2/4}/U\le \alpha^{1/4}/(cd)$. Thus $D=U/(c\alpha^{(i_1-1)/4})\ge d$.
\end{proof}
The estimation of $P(d,\alpha^{i/4}/U)$ can be accomplished using a Chernoff bound, with each estimation requiring $O(\log (1/\delta))$ queries. We need item~\ref{I3} to show that, with probability at least $1-\delta$, all the estimations up to~$i_1$ are correct. 

\begin{claim}\label{Claim2}
  The probability that all the estimations of $P(d,\alpha^{i/4}/U)$ for $i\le i_1$ have an additive error of at most $\Delta(\alpha)/8$ is at least $1-\delta$.
\end{claim}
\begin{proof}
    Since by item~\ref{I3}, $P(d,x)\le c'P(d,\alpha^{1/4}x)$, $c'<1$ and $P(d,\alpha^{i_1/4}/U)\le 1$, there is a constant $j_0$ such that $P(d,\alpha^{(i_1-j_0)/4}/U)\le \Delta(\alpha)/32$.
    Since $j_0$ and $\Delta(\alpha)/8$ are constants, we can estimate all $P(d,\alpha^{(i_1-j)/4}/L)$, $j=0,1,\ldots,j_0$ with additive error of at most $\Delta(\alpha)/8$ and probability at least $1-\delta/2$ with $O(\log (1/\delta))$ queries.

    Since $P(d,\alpha^{(i_1-j)/4}/U)\le \Delta(\alpha)/32$, by item~\ref{I3}, we have $P(d,\alpha^{(i_1-j-i)/4}/U)\le c'^i \Delta(\alpha)/32$, $i=1,2,\ldots,i_1-j$.
     By Lemma~\ref{Chernoff}, the probability that the estimation of $P(d,\alpha^{(i_1-j-i)/4}/U)$ has additive error greater than $\Delta(\alpha)/8$ is at most  $$\left(\frac{eP(d,\alpha^{(i_1-j-i)/4}/U)}{\Delta(\alpha)/8}\right)^{(\Delta(\alpha)/8)O(\log(1/\delta))}\le \left(\frac{ec'^i}{4}\right)^{O(\log(1/\delta))}=\left(\frac{\delta}{4}\right)^{i+1}.$$
    
    The probability that  one of the estimations of $P(d,1/U), P(d,\alpha^{1/4}/U),\ldots,$ $P(d,\alpha^{(i_1-j_0)/4}/U)$ has additive error greater than $\Delta(\alpha)/8$ is at most
    \begin{eqnarray*}
        \sum_{i=0}^{i_1-j_0} \left(\frac{\delta}{4} \right)^{i+1}\le  \frac{\delta}{2}.
    \end{eqnarray*}

\end{proof}

We now prove items~\ref{I1}-\ref{I3}.
Item~\ref{I1} is clear.

Since $P$ is a strictly monotone increasing function, if items~\ref{I2}-\ref{I3} hold for $\alpha=\alpha'$, they also 
hold for any $\alpha\ge \alpha'$. Therefore, it is sufficient to prove that items~\ref{I2}-\ref{I3} hold for any constant $1<\alpha\le 2$. 

Since $P(d,p)=1-(1-p)^d$, by the mean value theorem, Lemma~\ref{Ec} and since $\partial P(d,p)/\partial p=d(1-p)^{d-1}$ there is $1/(cd)\ge \eta\ge \alpha^{-1/2}/(cd)$ such that
\begin{eqnarray*}
  \Delta(\alpha)&=&P(d,1/(cd))-P(d,\alpha^{-1/2}/(cd))\\ &=&d(1-\eta)^{d-1}(1/(cd)-\alpha^{-{1/2}}/(cd))\\
    &\ge&\left(1-\frac{1}{cd}\right)^{d-1}(1-\alpha^{-1/2})/c\\
    &\ge& \left(1-\frac{1}{c(d-1)}\right)^{d-1}(1-\alpha^{-1/2})/c \\
    &\ge& (e^{-1/c}-A/(d-1))(1-\alpha^{-1/2})/c =O_d(1).
\end{eqnarray*}
This implies item~\ref{I2} for any constant $c$. 

We now prove item~\ref{I25}. Given a constant $\beta$, let $b=\lceil A/\beta\rceil$ where $A$ is the constant in Lemma~\ref{Ec}. By Lemma~\ref{Ec}, we have $$0\le P(b,1/(cb))-(1-e^{-1/c})\le \frac{A}{b}\le \beta.$$ 

We will now prove item~\ref{I3}. Let $c=\alpha$. Since for $0\le x\le 1$ we have $1-dx\le (1-x)^d\le 1-dx+d^2x^2/2$, it follows that
$P(d,x)=1-(1-x)^d\le dx$. Since $0\le x\le 1/cd$ and $c=\alpha$, we have $$P(d,\alpha^{1/4}x)=1-(1-\alpha^{1/4}x)^d\ge \alpha^{1/4}dx-d^2\alpha^{1/2}x^2/2\ge (\alpha^{1/4}-\alpha^{-1/2})dx\ge (\alpha^{1/4}-\alpha^{-1/2})P(d,x).$$

This completes the proof for $\lambda=1$ and any $\alpha=1+\Omega(1)$. 
Now, we will extend the proof to cover any constant $\lambda>1$ and $\alpha=1+\Omega(1)$.

It is easy to verify that
$$\frac{\partial P_\lambda(d,p)}{\partial p}={d\choose \lambda-1}(d-\lambda+1)p^{\lambda-1}(1-p)^{d-\lambda}.$$
To get the result, we show
\begin{enumerate}[label=L\arabic*]        \item\label{IL1}\hskip -0.08in . $P_\lambda(d,p)$ is a strictly monotone increasing function in $p$.
        \item\label{IL2}\hskip -0.08in . For some constant $c\ge 1$, we have $\Delta_\lambda(\alpha):=P_\lambda(d,\lambda/(cd))-P_\lambda(d,\alpha^{-1/2}\lambda/(cd))=O_d(1)$.
        \item\label{IL25}\hskip -0.08in . For every constant $\beta$ there is a constant $d'$ such that for every $d\ge d'$ we have $|P_\lambda(d,\lambda/(cd))-\lim_{x\to\infty}P_\lambda(x,\lambda/(cx))|\le \beta.$
        \item\label{IL3}\hskip -0.08in . There is a constant $c'<1$ such that for every $0\le x\le \lambda/cd$ we have $P_\lambda(d,x)\le c'P_\lambda(d,\alpha^{1/4}x)$.
    \end{enumerate}
The algorithm and its correctness are the same as the case of $\lambda=1$. Simply add $\lambda$ as a subscript to $P$ and replace $c$ with $c/\lambda$. So we only need to prove items \ref{IL1}-\ref{IL3} for $\lambda>1$.

Item \ref{IL1} follows because $\partial P_\lambda(d,p)/\partial p>0$ for all $0< p\le 1$. Since $P_\lambda$ is a strictly monotone increasing function, if items~\ref{IL2}-\ref{IL3} hold for $\alpha=\alpha'$, then they also 
hold for any $\alpha\ge \alpha'$. Therefore, it is enough to prove items~\ref{IL2} and~\ref{IL3} for any constant $1<\alpha\le 2$. 

Let 
$$c=\frac{2\lambda}{1-\alpha^{-1/4}}.$$
By the mean value theorem, there is $1/\alpha^{1/2}\le \eta\le 1$ such that 
\begin{eqnarray*}
    \Delta_\lambda(\alpha)&=& P_\lambda(d,\lambda/(cd))-P_\lambda(d,\lambda/(\alpha^{1/2}cd))\\
    &=& {d\choose \lambda-1} (d-\lambda+1)\left(\frac{\eta\lambda}{cd}\right)^\lambda\left(1-\frac{\eta\lambda}{cd}\right)^{d-\lambda}\left(\frac{\lambda}{cd}-\frac{\lambda}{\alpha^{1/2}cd}\right)\\
    &=&\frac{d(d-1)\cdots(d-\lambda+1)}{d^{\lambda+1}} \frac{\lambda^{\lambda+1}}{(\lambda-1)!}\frac{\eta^\lambda}{c^{\lambda+1}} \left(1-\frac{\eta\lambda}{cd}\right)^{d-\lambda} (1-\alpha^{-1/2})\\
    &\ge& 2^{\lambda+1}\frac{\lambda^{\lambda+1}}{(\lambda-1)!}\frac{\eta^\lambda}{c^{\lambda+1}}\frac{1}{2}(1-\alpha^{-1/2})
    \hskip 1in \lambda<d/2, c>2\lambda\\
    &=&O_d(1).
\end{eqnarray*}
This proves item~\ref{IL2}.

We now prove item~\ref{IL25}. First, since $\lambda$ is constant,
$$\lim_{x\to\infty}P_\lambda(x,\lambda/(cx))=1-\sum_{i=0}^{\lambda-1} \left(\frac{\lambda}{c}\right)^i\frac{e^{-\lambda/c}}{i!}.$$
For the proof we will use the following inequalities: For any $i<\lambda$ we have
$$1\ge \frac{d(d-1)(d-2)\cdots(d-i+1)}{d^i}\ge 1-\frac{\lambda^2}{2d}  \mbox{\ \ \ and\ \ \ } 1\ge \left(1-\frac{\lambda}{cd}\right)^i\ge 1-\frac{\lambda^2}{cd}.$$
Now, by Lemma~\ref{Ec} and the above inequalities,
\begin{eqnarray*}
 \left|P_\lambda(d,\lambda/cd)-\left(1-\sum_{i=0}^{\lambda-1} \left(\frac{\lambda}{c}\right)^i\frac{e^{-\lambda/c}}{i!}\right)\right|&=&\left|\sum_{i=0}^{\lambda-1}\left(\left(\frac{\lambda}{c}\right)^i\frac{e^{-\lambda/c}}{i!}-{d\choose i}\left(\frac{\lambda}{cd}\right)^i\left(1-\frac{\lambda}{cd}\right)^{d-i}\right)\right|\\
 &=& \left|\sum_{i=0}^{\lambda-1} \left(\left(\frac{\lambda}{c}\right)^i\frac{1}{i!}\left(e^{-\lambda/c}-\frac{\prod_{\ell=0}^{i-1}(d-\ell)}{d^i} \frac{\left(1-\frac{\lambda}{cd}\right)^d}{(1-\lambda/cd)^i}\right)\right)\right|\\
 &\le& \sum_{i=1}^{\lambda-1} \left(\frac{\lambda}{c}\right)^i\frac{1}{i!}
 \left|e^{-\lambda/c} -\left(1\pm O\left(\frac{1}{d}\right)\right)
 \left(1-\frac{\lambda}{cd}\right)^d\right| \\
 &\le& \sum_{i=1}^{\lambda-1} \left(\frac{\lambda}{c}\right)^i\frac{1}{i!}
 \left|e^{-\lambda/c} -\left(1-\frac{\lambda}{cd}\right)^d\pm O\left(\frac{1}{d}\right)\right| 
 \\ &=&O\left(\frac{1}{d}\right).
\end{eqnarray*}
This proves item~\ref{IL25}.

We now prove item~\ref{IL3}. For every $0\le j\le d-\lambda$ and $0\le x\le \lambda/cd$, we have
\begin{eqnarray*}
    \frac{{d\choose \lambda+j}x^{\lambda+j}(1-x)^{d-\lambda-j}}{{d\choose \lambda+j}(\alpha^{1/4}x)^{\lambda+j}(1-\alpha^{1/4}x)^{d-\lambda-j}}&\le&\frac{1}{\alpha^{\lambda/4}(1-\alpha^{1/4}x)^d}\\
    &\le& \frac{1}{\alpha^{\lambda/4}(1-\alpha^{1/4}xd)}\\
    &\le& \frac{1}{\alpha^{\lambda/4}(1-\alpha^{1/4}\lambda/c)}\\
    &\le& \frac{1}{\alpha^{(\lambda-1)/4}}.\hskip 1in c>\frac{\alpha^{1/4}\lambda}{1-\alpha^{-1/4}}
\end{eqnarray*}
Therefore, by~(\ref{PLam}), 
$$P_\lambda(d,x)=\sum_{i=\lambda}^{d} {d\choose i}x^i(1-x)^{d-i}\le\frac{1}{\alpha^{(\lambda-1)/4}} \sum_{i=\lambda}^{d} {d\choose i}(\alpha^{1/4}x)^i(1-\alpha^{1/4}x)^{d-i}\le \frac{1}{\alpha^{(\lambda-1)/4}}P_\lambda(d,\alpha^{1/4}x).$$
This implies item~\ref{IL3} and the result follows.
\end{proof}

\section{Exploring coupling: basic concepts}\label{CouplingBC}
In this appendix, we provide basic definitions and results for readers who may not be familiar with the coupling technique.  
To ensure completeness, we begin with fundamental concepts in probability.

Let $\cS$ be a set. We say that $\cF$ is {\it {$\sigma$}-algebra} on $\cS$ if $\cF\subseteq 2^\cS$ is a set of subsets of $\cS$ and $\cF$ is closed under complement, unions and countable intersections. 
The pair $(\cS,\cF)$ is called a {\it measurable space}. A {\it measure} on a measurable space $(\cS,\cF)$ is a function $\mu:\cF\to \mathbb{R}\cup\{-\infty,+\infty\}$ such that: (1) for every $F\in \cF$, $\mu(F)>0$. 
(2) $\mu(\emptyset)=0$ and (3) for any $\{F_i\}_{i\in \mathbb{N}}$, $F_i\in \cF$,
$\mu(\cup_{i=1}^\infty F_i)=\sum_{i=1}^\infty \mu(F_i)$. The tuple $(\cS,\cF,\mu)$ is called a {\it measure space}. If $\mu(\cS)=1$ then $\mu$ is called a {\it probability measure} and $(\cS,\cF,\mu)$ is called a {\it probability space}. 

Given two measurable spaces $(\cS_1,\cF_1)$ and $(\cS_2,\cF_2)$. A function $X:\cS_1\to \cS_2$ is a {\it  measurable function} if for every $F\in \cF_2$, $X^{-1}(F):=\{\omega\in \cS_1|X(\omega)\in F\}\in \cF_1$. When $(\cS_1,\cF_1,\mu)$ is a probability space, then $X$ is called a {\it random variable}. For such $X$, we say that $X$ is {\it defined on} $(\cS_1,\cF_1)$ and {\it takes values in} $(\cS_2,\cF_2)$.  
The {\it law of} $X$, denoted by $\mu_X$, is a probability measure on $(\cS_2,\cF_2)$ defined as: For every $F\in\cF_2$ $\mu_X(F)=\mu(X^{-1}(F))$.

Now, we give the formal definition of coupling.

\begin{definition}[Coupling]
    Let $\mu_1$ and $\mu_2$ be probability measures on the same measurable space $(\cS,\cF)$. A {\it coupling of $\mu_1$ and $\mu_2$} is a probability measure $\mu$ on the product space\footnote{Here $\cF\times \cF$ is the smallest $\sigma$-algebra that contains the Cartesian product of $\cF$ with itself.} $(\cS\times \cS,\cF\times \cF)$ such that for every $F\in \cF$, we have $\mu(F\times \cS)=\mu_1(F)$ and $\mu(\cS\times F)=\mu_2(F)$.  
\end{definition}

For two random variables $X$ and $Y$ taking values in $(\cS,\cF)$ (but not necessarily defined on the same probability space), a coupling of $X$ and $Y$ is a joint variable $(X',Y')$ taking values in $(\cS\times \cS,\cF\times \cF)$ where\footnote{$\mu_{(X',Y')}$ is the law of ${(X',Y')}$.} $\mu_{(X',Y')}$ is a coupling of $\mu_X$ and $\mu_Y$. We also say that $(X',Y')$ is a coupling of $\mu_1$ and $\mu_2$ if $\mu_{(X',Y')}$ is a coupling of of $\mu_1$ and $\mu_2$.

\subsection{Preliminary results}
In this section, we present two well-known results that will be used in the paper, along with their proofs for completeness.

\begin{lemma}\label{Composition}
    Let $(\cS,\cF,\mu)$ be a probability space, $(\cS_1,\cF_1)$ and $(\cS_2,\cF_2)$ measurable spaces, $X',Y':\cS\to\cS_1$ random variables, $X$ and $Y$ are random variables that take values in $(\cS_1,\cF_1)$ and $f:\cS_1\to\cS_2$ a measurable function such that $f(\cS_1)=\cS_2$. If $(X',Y')$ is a coupling of $X$ and $Y$ then $(f(X'),f(Y'))$ is a coupling of $f(X)$ and $f(Y)$.
\end{lemma}
\begin{proof}
    It is clear that $f(X'),f(Y'):\cS\to \cS_2$ are random variables that take values in $(\cS_2,\cF_2)$ and $f(X)$ and $f(Y)$ take values in $(\cS_2,\cF_2)$. Now for any $F\in \cF_2$,
    \begin{eqnarray*}
        \mu_{(f(X'),f(Y'))}(F\times \cS_2)&=&\mu(\{\omega\in \cS|f(X'(\omega))\in F, f(Y'(\omega))\in \cS_2\})\\
        &=& \mu(\{\omega\in \cS|X'(\omega)\in f^{-1}(F), Y'(\omega)\in f^{-1}(S_2)\})\\
        &=&\mu_{(X',Y')}(f^{-1}(F)\times f^{-1}(\cS_2))
        \\
        &=&\mu_{(X',Y')}(f^{-1}(F)\times \cS_1)=\mu_X(f^{-1}(F))=\mu_{f(X)}(F).
    \end{eqnarray*}
    In the same way $\mu_{(f(X'),f(Y'))}(\cS_2\times F)=\mu_{f(Y)}(F)$.
\end{proof}

\begin{lemma}
    Let $\mu_1$ and $\mu_2$ be two probability measures on the measurable space $(\cS,\cF)$. For any coupling $(X,Y)$ of $\mu_1$ and $\mu_2$ we have
    \[
    \|\mu_1 - \mu_2\|_{\TV} \le \Pr([X \ne Y])
    \]
    where\footnote{Here we also assume that $[X\ne Y]\in \cF'$} $[X\ne Y]=\{\omega\in \cS'|X(\omega)\not= Y(\omega)\}$ and $X$ and $Y$ are random variables defined on the probability space $(\cS',\cF',\Pr)$ (and take values in $(\cS,\cF)$).
\end{lemma}
\begin{proof} 
Let $[X\ne Y]=\{\omega\in \cS|X(\omega)= Y(\omega)\}$. Since for any $B \in \cF$ we have $X^{-1}(B)\cap [X=Y]=Y^{-1}(B)\cap [X=Y]$,
    for any $A\in \cF$, we have
    \begin{align*}
    \mu_1(A) - \mu_2(A) &= \mu_{X\times Y}(A \times\cS) - \mu_{X\times Y}(\cS\times A) \\
    &=\Pr(\{\omega\in S'|X(\omega)\in A, Y(\omega)\in \cS\})- \Pr(\{\omega\in S'|X(\omega)\in \cS, Y(\omega)\in A\}) \\
    &=\Pr(X^{-1}(A))- \Pr(Y^{-1}(A)) \\
    &= \Pr(X^{-1}(A)\cap [X = Y]) + \Pr(X^{-1}(A)\cap [X \ne Y])  \\&\hskip .5in - \Pr(Y^{-1}(A)\cap [X = Y])- \Pr(Y^{-1}(A)\cap [X \ne Y]) \\
    &= \Pr(X^{-1}(A)\cap [X \ne Y]) - \Pr(Y^{-1}(A)\cap [X \ne Y]) \\
    &\le \Pr([X \ne Y]). \qedhere
\end{align*} 
\end{proof}

For a more extensive treatment of coupling, see \cite{den2012probability} or \cite[Chapter 4]{roch2015modern}.

\section{Randomized query algorithm to deterministic distinguisher}\label{EandToDet}
In this section, we establish the connection between randomized query algorithm for group testing estimation and deterministic distinguisher of the hard distributions constructed in \cref{sec:proof}.
\begin{lemma}
    If there is a randomized non-adaptive algorithm that asks $q$ queries and for every set of defective items $I$, with probability at least $2/3$, outputs an $\alpha$-estimation of $|I|$, then there are $q$ queries $\vec Q^{\ge \lambda}=(Q_1,\ldots,Q_q)$ and a deterministic distinguisher $\Psi$ between $\vec Q^{\ge}(\muEven)$ and $\vec Q^{\ge}(\muOdd)$. 

    In particular, (by Lemma~\ref{lem:TV-algo})
    $$\|\vec{Q}^{\geq\lambda}(\muEven) -\vec{Q}^{\geq\lambda}(\muOdd)\|_{\TV} \ge \frac{1}{3}. $$
\end{lemma}
\begin{proof}
    Let ${\cal A}(s,I)$ be a randomized non-adaptive algorithm that asks $q$ queries and for any set of defective items $I$, with probability at least $2/3$, outputs an $\alpha$-estimation of $|I|$, where $s$ is the random seed of the algorithm. 
    
    Consider the indicator random variable $X(s,I)$, which equals $1$ if the estimation is correct. For every $I$, we have $E_s[X(s,I)]\ge 2/3$. Consider the distribution $D$ of $I$ where, with probability $1/2$, $I$ is chosen according to $\muOdd$, and with probability $1/2$, it is chosen according to $\muEven$. Then $$E_s[E_{I\sim D}[X(s,I)]]=E_{I\sim D}[E_s[X(s,I)]]\ge 2/3$$ and therefore there exists an $s_0$ such that $E_{I\sim D}[X(s_0,I)]\ge 2/3$. In other words, there exist a set of $q$ queries $\vec Q^{\ge \lambda}=(Q_1,\ldots,Q_q)$ such that, for random $I$ chosen according to $D$, with probability at least $2/3$, the deterministic algorithm ${\cal A}(s_0,I)$ estimates correctly $|I|$. The same algorithm can distinguish between $\vec Q^{\ge}(\muEven)$ and $\vec Q^{\ge}(\muOdd)$ because, by the choice of $\beta$, there is exactly one $s\in\SOdd\cup\SEven$ which the estimation algorithm output is within the range of $[s,\alpha s]$.
\end{proof}
\end{document}